\documentclass[journal]{IEEEtran}
\IEEEoverridecommandlockouts

\usepackage{cite}
\usepackage{amsmath,amssymb,upgreek,fge,mathabx}
\usepackage{algorithmic}
\usepackage{textcomp}
\usepackage{xcolor}
\usepackage{cite}
\usepackage{url}
\usepackage{amssymb}
\usepackage{euscript}
\usepackage{amssymb}
\usepackage{pifont}

\usepackage{bbding}
\usepackage{graphicx}
\usepackage{multirow}
\usepackage{algorithm}
\usepackage{algorithmic}
\usepackage{ipa}
\usepackage[tone,extra,safe]{tipa}
\usepackage{changepage}
\usepackage{color}
\usepackage{caption}
\usepackage{subcaption}
\usepackage{makecell} 
\usepackage{todonotes}
\usepackage{scalerel}
\usepackage{tikz}
\usetikzlibrary{svg.path}
\usepackage{balance}
\usepackage{svg}
\usetikzlibrary{fit, positioning, shapes.geometric, fit, arrows.meta, calc, backgrounds}
\usepackage{pgfplots}
\pgfplotsset{compat=newest}

\usepackage{etoolbox} 
\usepackage{xcolor} 

\let\mybibitem\bibitem
\renewcommand{\bibitem}[1]{%
  \color{black}%
  \ifstrequal{#1}{9978646}{\color{black}\mybibitem{#1}}{%
    \ifstrequal{#1}{9968296}{\color{black}\mybibitem{#1}}{%
        \ifstrequal{#1}{10238831}{\color{black}\mybibitem{#1}}{%
            \ifstrequal{#1}{10065563}{\color{black}\mybibitem{#1}}{%
                \ifstrequal{#1}{10375242}{\color{black}\mybibitem{#1}}{%
                    \mybibitem{#1}%
                    }%
            }%
       }%
    }%
  }%
}

\definecolor{darkgreen}{rgb}{0, 0.5, 0} 
\definecolor{lightpurple}{rgb}{0.7, 0.4, 1} 

\definecolor{orcidlogocol}{HTML}{A6CE39}
\tikzset{
orcidlogo/.pic={
\fill[orcidlogocol] svg{M256,128c0,70.7-57.3,128-128,128C57.3,256,0,198.7,0,128C0,57.3,57.3,0,128,0C198.7,0,256,57.3,256,128z};
\fill[white] svg{M86.3,186.2H70.9V79.1h15.4v48.4V186.2z}
svg{M108.9,79.1h41.6c39.6,0,57,28.3,57,53.6c0,27.5-21.5,53.6-56.8,53.6h-41.8V79.1z M124.3,172.4h24.5c34.9,0,42.9-26.5,42.9-39.7c0-21.5-13.7-39.7-43.7-39.7h-23.7V172.4z}
svg{M88.7,56.8c0,5.5-4.5,10.1-10.1,10.1c-5.6,0-10.1-4.6-10.1-10.1c0-5.6,4.5-10.1,10.1-10.1C84.2,46.7,88.7,51.3,88.7,56.8z};
}
}
\newcommand\orcidicon[1]{\href{https://orcid.org/#1}{\mbox{\scalerel*{
\begin{tikzpicture}[yscale=-1,transform shape]
\pic{orcidlogo};
\end{tikzpicture}
}{|}}}}

\usepackage[colorlinks=true,linkcolor=blue,citecolor=blue]{hyperref}
\newtheorem{proof}{Proof}
\newtheorem{proposition}{\bf{Proposition}}

\newtheorem{corollary}{\bf{Corollary}}
\newtheorem{remark}{Remark}

\usepackage[most]{tcolorbox}

\newcommand*{\QED}[1][$\blacksquare$]{%
\leavevmode\unskip\penalty9999 \hbox{}\nobreak\hfill
\quad\hbox{#1}%
}

\usepackage{flushend}

\begin{document} 

\title{\huge Miniature UAV-Aided Cooperative THz Networks with Reconfigurable Energy Harvesting Holographic Surfaces}

\author{


Yifei Song, 
Jalal Jalali, \IEEEmembership{Member, IEEE}, 
Yanyu Qin,
Mostafa Darabi, \IEEEmembership{Member, IEEE}, 
Filip Lemic, \IEEEmembership{Member, IEEE}, and\
Natasha Devroye, \IEEEmembership{Fellow, IEEE}

\thanks{Yifei Song and Yanyu Qin are affiliated with the Bradley Department of Electrical and Computer Engineering and the Department of Computer Science at Virginia Tech, Blacksburg, VA, USA, respectively.
Jalal Jalali and Mostafa Darabi are with the Wireless Communication Research Group, JuliaSpace LLC., Chicago, IL, USA.
Filip Lemic is affiliated with the AI-Driven Systems Lab, i2Cat Foundation, Spain, and the Faculty of Electrical Engineering and Computing, University of Zagreb, Croatia.
Natasha Devroye is affiliated with the Department of Electrical and Computer Engineering at the University of Illinois at Chicago, Chicago, IL, USA.
The corresponding author is Jalal Jalali (josh@juliaspace.com).
}
\thanks{This work was supported by the Belgian American Educational Foundation (BAEF), and in part, by the U.S. National Science Foundation under Grant CNS-2225511. 
}
}

\maketitle
\begin{abstract}
This paper focuses on enhancing the energy efficiency (EE) of a cooperative network that features a miniature unmanned aerial vehicle (UAV) operating at terahertz (THz) frequencies and equipped with holographic surfaces to improve network performance. Unlike traditional reconfigurable intelligent surfaces (RIS), which serve as passive relays for signal reflection, this work introduces a novel concept: energy harvesting (EH) using reconfigurable holographic surfaces (RHS). These surfaces provide more powerful and focused energy delivery during wireless power transfer than RIS and are mounted on the miniature UAV. In this system, a source node enables the UAV to simultaneously receive both information and energy signals, with the harvested energy powering data transmission to a specific destination. The EE optimization problem involves adjusting non-orthogonal multiple access (NOMA) power coefficients and the UAV's flight path while accounting for the unique characteristics of the THz channel. The problem is solved in two stages to maximize EE and meet a target transmission rate. \textcolor{black}{The UAV trajectory is optimized using a successive convex approximation (SCA) method, followed by the adjustment of NOMA power coefficients through a quadratic transform technique.} Simulation results demonstrate the effectiveness of the proposed algorithm, showing significant improvements over baseline methods.
\end{abstract}
\begin{IEEEkeywords}
Cooperative communication, energy efficiency (EE), energy harvesting (EH), reconfigurable holographic surfaces (RHS), and miniature unmanned aerial vehicles (UAV).
\end{IEEEkeywords}

\section{Introduction}
\indent 
\IEEEPARstart{B}{uilding} on the original vision of 5G and extending into 6G, future wireless networks are expected to support enhanced mobile broadband (eMBB), ultra-reliable and low latency communications (URLLC) massive machine types communications (mMTC) with the requirements of improving data rates, improving network capacity, reducing latency, and minimizing energy consumption~\cite{10077210}.
Researchers are investigating new network topologies focusing on establishing extensive backhaul links to boost network capacity~\cite{9022993}. 
Unmanned aerial vehicles (UAVs) have gained significant attention due to their unique advantage of establishing line-of-sight (LOS) communication links, enabling them to provide services to ground users while meeting quality of service (QoS) requirements.
Specifically, miniature UAVs \cite{10423569}, including nano-UAVs—characterized by palm-sized dimensions and weights below 250 grams—and micro-UAVs, which are backpack-portable and weigh between 250 grams and 20 kilograms, are particularly well-suited for operations in confined spaces due to their agility and compact design \cite{10113154}. In contrast to standard UAVs, which may exceed 150 kilograms, miniature UAVs can safely operate in close proximity to humans \cite{10423569}, making them ideal for indoor applications~\cite{9440350}. These capabilities enable advanced use cases, such as deployment in industrial indoor rich-scattering environments~\cite{9600850},
precise environmental monitoring~\cite{9817083}, emergency search-and-rescue in collapsed structures~\cite{9440534}, and enhanced immersive virtual reality~\cite{dang2022low} experiences through real-time mapping and interaction.

Simultaneously, ensuring massive connectivity within the mMTC framework remains a significant challenge, as accommodating trillions of devices within the already congested and limited sub-6 GHz spectrum becomes increasingly difficult. Moreover, the rapid growth of the Internet of Things (IoT), with millions to billions of ubiquitously connected devices, presents even more severe challenges in supporting such large-scale connectivity. To address these limitations, a shift towards higher-frequency terahertz (THz) communication is being actively explored, offering the potential for data rates in the hundreds of gigabits per second.
Additionally, non-orthogonal multiple access (NOMA)~\cite{9022993, 9340353} is gaining traction as a method to support multiple users simultaneously on the same frequency and time slots, using efficient interference cancellation techniques.
Together, these advancements aim to meet the growing demands of next-generation communication networks.
\begin{table*}[ht]
\caption{Overview of RIS- and RHS-assisted UAV works with focus on EH, objectives, and optimization approaches.}
\label{table:lit}
\centering
\begin{tabular}{|l|l|l|l|l|l|}
\hline
Ref.                                   & EH         & RHS        & Objectives                                                                                      & Optimization Parameters                                                  & Algorithms                                        \\ \hline
\cite{9989438}             & \XSolid    & \XSolid    & \makecell[l]{Enhance EE by optimizing UAV trajectory \\ and number of RIS elements}             & \makecell[l]{UAV weight, RIS weight, \\ coverage probability}           & \makecell[l]{ALOHA \& Code \\ Combining MAC protocol}             \\ \hline
\cite{9645164}             & \XSolid    & \XSolid    & \makecell[l]{Enhance EE via UAV power allocation \\ and RIS phase shift optimization}           & \makecell[l]{Power allocation, \\ phase-shift matrix}                   & DRL                                               \\ \hline
\cite{kumar2024maximizing} & \Checkmark & \XSolid    & \makecell[l]{Enhance EE using time-space RIS-assisted \\ EH scheme}                             & \makecell[l]{QoS, trajectory, \\ resource allocation}                   & DRL                                               \\ \hline
\cite{9771999}             & \Checkmark & \XSolid    & \makecell[l]{Maximize harvested energy using RIS-assisted \\ EH while maintaining QoS}          & \makecell[l]{Passive reflect-arrays, \\ resource allocation}            & DRL                                               \\ \hline
\cite{10274676}            & \Checkmark & \XSolid    & \makecell[l]{Maximize bidirectional throughput in \\ RIS-assisted UAV networks}                 & \makecell[l]{Time allocation, transmit power, \\ EH ratio, trajectory}  & \makecell[l]{Block coordinate descent \\ \& unsupervised learning} \\ \hline
\cite{10399860}            & \Checkmark & \XSolid    & \makecell[l]{Enhance end-to-end throughput in \\ UAV-enabled EH relay network}                  & \makecell[l]{Throughput, transmit power, \\ trajectory}                 & \makecell[l]{Block coordinate descent \\ \& Lagrange duality}      \\ \hline
\cite{9635669}             & \Checkmark & \XSolid    & \makecell[l]{Evaluate coverage with \\ non-linear EH model}                                     & \makecell[l]{Coverage probability, trajectory, \\ transmission power, harvesting energy} & Stochastic geometry                               \\ \hline
\cite{sheemar2025joint}   & \XSolid    & \Checkmark & \makecell[l]{Maximize sum rate in \\ RHS-assisted UAV communication}                            & \makecell[l]{UAV position, digital beamforming, \\ holographic beamforming weights} & Gradient-ascent                                  \\ \hline
\cite{9696209}            & \XSolid    & \Checkmark & \makecell[l]{Optimize the sum rate for \\ multi-user communications}                            & \makecell[l]{Digital \& holographic beamforming,\\ receive combining}   & Coordinate ascent                                  \\ \hline
\cite{9848831}            & \XSolid    & \Checkmark & \makecell[l]{Maximize the sum rate using holographic-assisted \\ beamforming for LEO satellite communications} & \makecell[l]{Digital \& holographic beamforming} & Dynamic programming \\ \hline
\textbf{This work}         & \Checkmark & \Checkmark & \makecell[l]{Enhance EE and transmission rate in \\ holographic EH UAV networks}                & \makecell[l]{Power allocation parameters, \\ UAV trajectory}            & \makecell[l]{SCA }                  \\ \hline 
\end{tabular}
\end{table*}

\textcolor{black}{
More recently, the integration of UAVs with reconfigurable intelligent surfaces (RIS) has emerged as a promising approach to enhance wireless communication by utilizing the high mobility of UAVs and the ability of RIS to control signal reflections~\cite{10530340}. RIS technology, typically used as passive relays, adjusts the phase of incoming signals to improve coverage and capacity without requiring active amplification. However, RIS remains fundamentally limited by its reliance on external power supplies, control circuits, and its inability to sustain operation autonomously.
}

\textcolor{black}{
\textit{This paper builds on this concept by introducing Reconfigurable Holographic Surfaces (RHS) with integrated Energy Harvesting (EH)~\cite{jalali2025shape}}, thereby extending the RIS paradigm. To the best of our knowledge, we are the first to propose UAV–RHS integration, where RHS not only supports signal manipulation but also enables wireless power transfer for energy harvesting.
}
\textcolor{black}{
Compared with conventional RIS, RHS offers several fundamental advantages for UAV-assisted networks:
(i) \textbf{Energy neutrality:} RHS can harvest RF energy to power its own embedded sensors and control circuitry, eliminating the need for dedicated power sources, amplifiers, RF chains, or mixers~\cite{9136592}. This self-sustaining design is particularly critical for miniature UAVs, which are highly constrained by payload and limited battery capacity.
(ii) \textbf{Lightweight and compact design:} RHS employs ultra-thin, low-cost metasurfaces that can be seamlessly integrated onto UAVs or deployed on walls, ceilings, and industrial structures~\cite{9690474}, unlike RIS implementations that often require larger form factors and auxiliary circuitry.
(iii) \textbf{Scalability and deployment flexibility:} The lightweight and hardware-efficient architecture of RHS enables dense and practical deployment in diverse scenarios, including warehouses, factories, and even wearable applications~\cite{8741198}.
}

\textcolor{black}{
These advantages make RHS particularly attractive for UAV-assisted systems, where both flight endurance and communication performance are critical. By harvesting energy from impinging signals, RHS alleviates UAVs’ reliance on external energy sources while simultaneously enhancing coverage and communication quality of service (QoS). This combination of energy-neutral operation, compact design, and deployment flexibility positions RHS as a feasible and transformative technology for maximizing the energy efficiency (EE) of THz UAV-assisted networks.
}
The key contributions are as follows:
\begin{itemize}
    \item To the best of our knowledge, this is the first work to explicitly integrate RHS for EH within a novel THz NOMA cooperative communication model. In this model, the source node transmits: (1) a superimposed message to both the UAV and the destination node using distinct power allocation coefficients, as per NOMA principles; and (2) a dedicated power signal to the reconfigurable holographic surface (RHS) for EH, thereby extending its operational lifetime.
    \item An optimization problem is formulated and solved to optimize the NOMA power allocation coefficients and the UAV trajectory in a three-dimensional (3D) system, aiming to maximize both the system's EE and the target transmission rate.
\end{itemize}
The paper is organized as follows: Section \ref{section_2} introduces the related work. Section \ref{section_3} presents the system model and optimization problem. Section \ref{section_4} covers the solution approach. Section \ref{chap8_sec6} discusses simulation results, and Section \ref{sec_5} concludes with future research directions.


In this paper, the notation is as follows: non-bold lowercase letters \(a\) denote scalars, bold lowercase letters \(\mathbf{a}\) represent vectors, bold uppercase letters \(\mathbf{A}\) represent matrices, and calligraphic letters \({\mathcal{A}}\) denote tensors. The symbol \((\cdot)^{\mathrm{T}}\) indicates the transpose operation. The set of real numbers is represented by \(\mathbb{R}\). For the convenience of readers, the notations used throughout this paper are summarized in Table~\ref{table:notation_updated}.
\begin{figure*}[ht]
    \centering
    \begin{subfigure}[b]{0.49\textwidth}
        \centering
        \includegraphics[width=\textwidth]{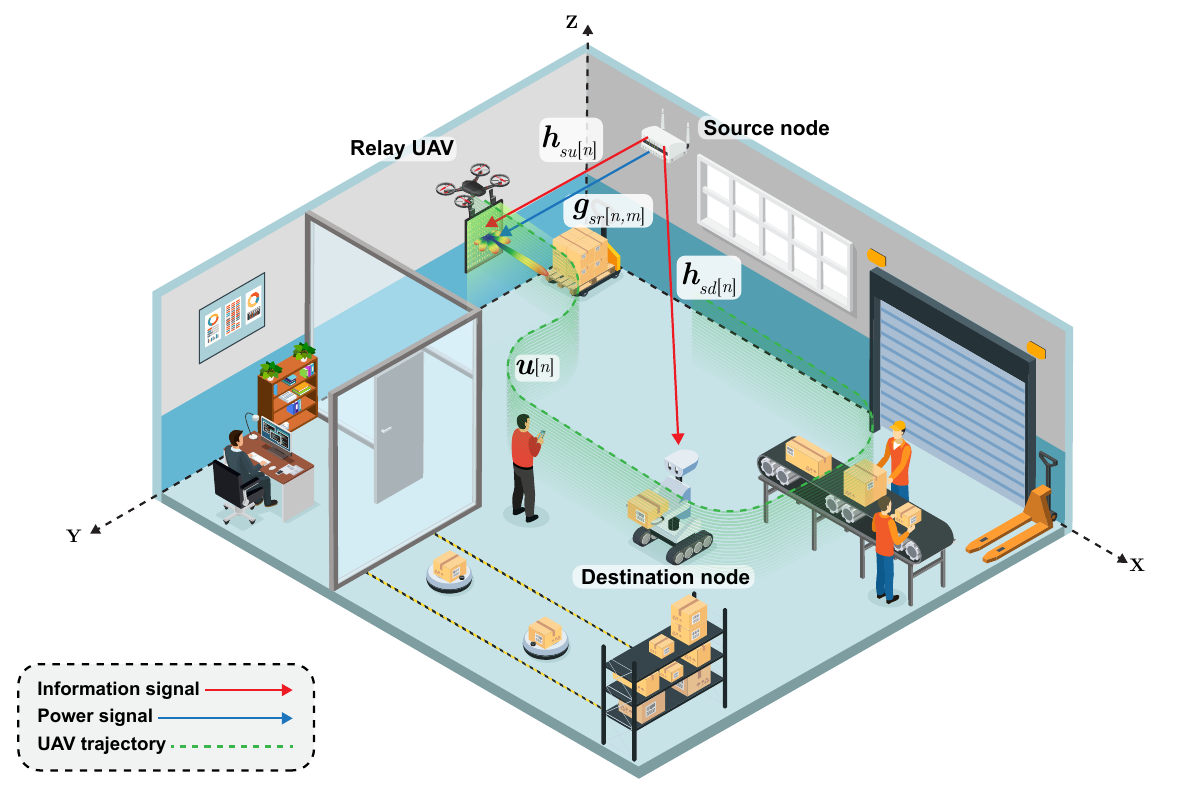}
        \caption{Episode 1: Direct Transmission and EH with RHS.}
        \label{fig1:subfig1}
    \end{subfigure}
    \hfill
    \begin{subfigure}[b]{0.49\textwidth}
        \centering
        \includegraphics[width=\textwidth]{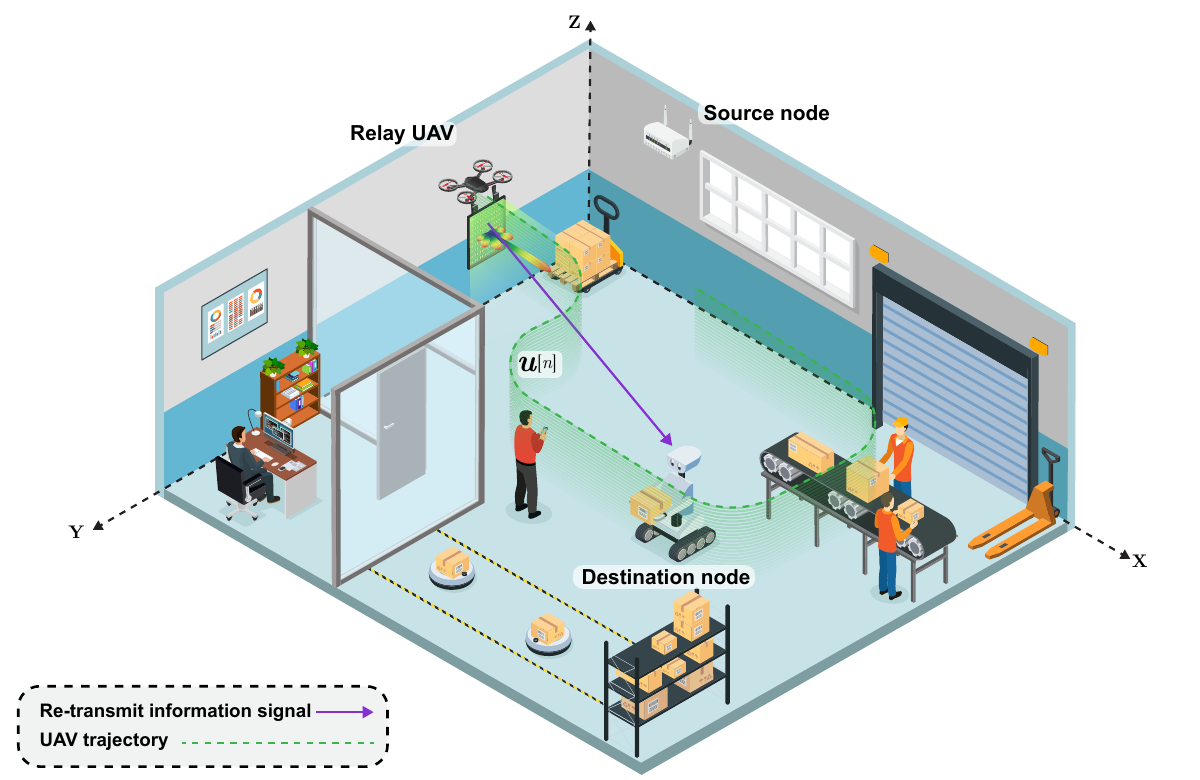}
        \caption{Episode 2: Cooperative Transmission.}
        \label{fig1:subfig2}
    \end{subfigure}
    \caption{Illustration of cooperative transmission by miniature UAVs in a THz network empowered by EH RHS. (a) Episode 1: Direct information transmission from the source node to the destination node, information transmission from the source node to the relay UAV, and power signal transmission from the source node to the UAV. (b) Episode 2: Information re-transmission from the relay UAV to the destination node using the energy harvested via EH RHS.}
    \label{fig1}
\end{figure*}

\begin{table}[ht]
\caption{Summary of Notation}
\label{table:notation_updated}
\centering
\setlength{\tabcolsep}{4pt} 
\begin{tabular}{|l|p{5cm}|}  
\hline
\textbf{Symbol} & \textbf{Description} \\
\hline
$\boldsymbol{u}[n]$ & UAV's 3D position at time slot $n$ \\
$\boldsymbol{s}[n]$ & Source node position at time slot $n$ \\
$\boldsymbol{d}[n]$ & Destination node position at time slot $n$ \\
$\boldsymbol{r}[n,m]$ & Position of $m$-th RHS element at slot $n$ \\
$T$ & Total UAV operation time \\
$N$ & Number of time slots \\
$V_{\max}$ & Maximum UAV speed \\
$\Delta_t$ & Duration of each time slot \\
$h_{su}[n]$, $h_{ud}[n]$, $h_{sd}[n]$ & Channel gains from source-to-UAV, UAV-to-destination, and source-to-destination \\
$g_{sr}[n,m]$ & Channel gain from source to $m$-th RHS element \\
$\xi(f)$ & Molecular absorption coefficient \\
$\fgeeszett_0$ & Reference path gain\\
$s[n]$ & Transmitted NOMA signal at time slot $n$ \\
$s_1[n]$, $s_2[n]$ & Information symbols intended for UAV and destination, respectively \\
$\fgelb_1[n]$, $\fgelb_2[n]$ & Power allocation coefficients for NOMA signals \\
$\fged_{\rm peak}$, $\fged_{\max}$ & Max instantaneous and average source transmit powers \\
$\varsigma_1[n]$, $\varsigma_2[n]$ & Power split ratios for ID and EH antennas \\
$\textdotbreve{a}[m]$, $\textdotbreve{a}_M$ & Absorption coefficient of $m$-th RHS element and uniform coefficient \\
$\omega[m]$, $\omega_M$ & Phase shift of $m$-th RHS element and uniform phase shift \\
$\fgef_{d \leftarrow u}^{(1)}[n]$, $\fgef_r^{(1)}[n]$ & SINR at UAV decoding destination's and own signals in Episode 1 \\
$\fgef_d^{(1)}[n]$, $\fgef_d^{(2)}[n]$, $\fgef_d^{\rm MRC}[n]$ & SINR at destination in Episodes 1, 2, and combined by MRC \\
$\mathcal{E}[n]$ & Energy harvested by RHS at time slot $n$ \\
$\fged[n]$, $\fged_{\rm EH}$ & Minimum required harvested power and harvested power at EH RHS \\
$\eta$ & Energy harvesting efficiency \\
$\mathcal{T}[n]$ & Transmission fraction of Episode 1 at time slot $n$ \\
${\fged_{\rm RHS}}[n]$ & UAV transmit power in Episode 2 powered by harvested energy \\
$\fged_c$ & UAV circuit power consumption \\
$\fged_{\rm sum}[n]$ & Total system power consumption at time slot $n$ \\
$R_{\rm sum}[n]$ & Sum rate at time slot $n$ \\
$\eta_{\rm EE}[n]$ & Energy efficiency at time slot $n$ \\
${\fgef_{\min}}[n]$, ${\gamma_{\min}}[n]$ & Minimum SINR thresholds at UAV and destination \\
${\varrho_1}^{(1)}[n]$, ${\varrho_2}^{(2)}[n]$ & Additive noise at destination during Episodes 1 and 2 \\
$\varepsilon_1^2[n]$, $\varepsilon_2^2[n]$ & Noise power at destination during Episodes 1 and 2 \\
$\boldsymbol{\alpha}$, $\boldsymbol{\beta}$ & Auxiliary parameter vectors for fractional objective transformation \\
\hline
\end{tabular}
\end{table}

\section{Related Work} \label{section_2}
Some studies have proposed energy-aware UAV-RIS models to enhance the number of tasks completed per flight \cite{9989438}. 
For example, Nguyen et al.~\cite{9645164} formulate an EE maximization problem and employ deep reinforcement learning to jointly optimize UAV power allocation and the RIS phase shift matrix.
Similarly, Kumar et al. \cite{kumar2024maximizing} and Peng et al. \cite{10051712} address power allocation with EH in a RIS-assisted UAV network under a dynamic wireless environment, using a deep reinforcement learning framework to enhance energy efficiency.
Xiao et al. \cite{10092842} introduced a solar-powered UAV-mounted RIS that provides external propulsion power and maximizes EE by optimizing the UAV trajectory alongside the beamforming active states.
Tyrovolas et al. \cite{10348506} studied a harvest-and-reflect (HaR) protocol designed to harvest energy for information transmission.
Lyu et al. \cite{9214497} explored a hybrid access point that transfers energy to both the RIS and users, enabling self-sustainable information transmission following the EH process. 
There are also inherent challenges when considering the use of THz for EH in UAV-RHS systems.
From a materials science perspective \cite{zhou2022patterned}, THz waves are readily absorbed by materials, particularly biological substances that resonate at THz frequencies. This effect is especially significant in polar molecules, e.g., water, where THz radiation induces dipole moments, enhancing absorption. Consequently, the high absorptivity of THz waves poses challenges for their use in wireless communication. 
We hypothesize that if RHS for EH at THz frequencies is developed in the future, it could be made from materials that leverage this absorptive property of THz waves to enhance EH efficiency. To our knowledge, this has not yet been done.


Although there have been advancements~\cite{kumar2024maximizing, 9771999, 10274676, 10399860, 9635669} considering EH in UAV networks, none have focused on RHS for this purpose. Similarly, works~\cite{sheemar2025joint, 9696209, 9848831} have considered RHS for various applications, but not for EH. Table~\ref{table:lit} presents a comparative overview of recent UAV communication studies assisted by RIS and RHS. It outlines whether EH and RHS are incorporated, summarizes the main objectives, lists key optimization variables, and details the algorithms adopted in each work. This overview helps contextualize the unique contributions and methodologies of the current work in relation to the existing literature. In this paper, we harness the absorptive properties of the THz spectrum by equipping a miniature UAV with RHS for EH. This strategy extends battery life during data transmission and introduces a novel cooperative communication framework for air-to-ground transmission.


\section{System Model and Problem formulation}\label{section_3}
We study a downlink NOMA transmission scenario within a miniature UAV-assisted RHS cooperative framework. The cooperative communication unfolds over two episodes. In the first episode, the source node employs NOMA to simultaneously transmit to both the destination and the miniature UAV. During this episode, the UAV performs EH via the RHS while decoding the source's information, and the destination directly receives its data. In the second episode, the UAV acts as an aerial relay, forwarding the decoded data to the destination using the energy harvested in the first episode.

In Fig. \ref{fig1}, the source node communicates with two terminals: a miniature UAV and a receiver destination node.  
The UAV serves as an EH-RHS to guarantee the high rate requirement of the destination node.  
A 3D coordinate system is utilized, where the source and destination are positioned at 
$\boldsymbol{s}(t) = {[s_{x}(t),s_{y} (t),H_s]}^\text{T} \in {\mathbb{R}^{3 \times 1}}$ and $\boldsymbol{d}(t) = {[d_{x}(t), d_{y}(t),H_d]^\text{T}} \in {\mathbb{R}^{3 \times 1}}$, respectively.
The destination node remains static on the ground, while the altitude of UAV with RHS and the source maintain fixed altitudes, though they differ from one another, i.e., $H_u=H_r$ (altitude of RHS) and $H_s$. 
At any given time $0 < t < T$, the UAV’s instantaneous position is represented as $\boldsymbol{u}(t) = {[x(t),y(t),H_u]^\text{T}} \in {\mathbb{R}^{3 \times 1}}$. 
Moreover, the coordinates of the RHS-equipped UAV are expressed by $\boldsymbol{r}(t,m) = {[x(t,m),y(t,m),H_u]^\text{T}} \in {\mathbb{R}^{3 \times 1}}$, where $m=\{1,\dots,M\}$ refers to the index of each holographic element.  
The total flight time of the UAV, denoted as $T$, is divided into $N$ equal time slots,  with the trajectory at each time slot denoted as $\boldsymbol{u}[n], \forall n \in \{ 1,...,N\}$. 
Each slot is small enough to treat the UAV position as nearly constant.  
The UAV's position and speed are subject to the following constraints:
\begin{subequations} 
\begin{align}
& ~~\boldsymbol{u}[1] = {\boldsymbol{u}_s}, 
\label{1a}\\
& ~~\boldsymbol{u}[N + 1] = {\boldsymbol{u}_e}, 
\label{1b}\\
& \left\| {\boldsymbol{u}[n + 1] - \boldsymbol{u}[n]} \right\| \le \Delta_t {V_{\max}} ,\; \forall n,
\label{1c}
\end{align}
\end{subequations}
where $V_{\max}$ is the maximum allowable speed, $\Delta_t$ denotes the length of each time slot, and $\boldsymbol{u}_s$ and $\boldsymbol{u}_e$ represent the UAV’s starting and ending positions, respectively.
The channel coefficients for source-to-UAV and UAV-to-destination are $h_{su}[n]$ and $h_{ud}[n]$, which adhere to the free-space path loss model and are expressed as:
\begin{equation}\label{h11}
h_{su}[n]= 
\frac{\fgeeszett _0}
{{{\left\| {\boldsymbol{u}[n] - \boldsymbol{s}[n]} \right\|}}}
e^{-\frac{\xi(f)}{2}{\left\| {\boldsymbol{u}[n] - \boldsymbol{s}[n]} \right\|}},\forall n,
\end{equation}
\begin{equation}\label{h12}
h_{ud}[n] = 
\frac{\fgeeszett_0}{\left\|\boldsymbol{u}[n] - \boldsymbol{d}[n] \right\|}
e^{-\frac{\xi(f)}{2}{\left\| {\boldsymbol{u}[n] - \boldsymbol{d}[n]} \right\|}},\forall n,
\end{equation}
and the channel gain between the source-to-RHS is given by: 
\begin{equation}\label{h13}
\hspace{-2mm}
g_{sr}[n,m] = 
\frac{\fgeeszett _0}{\left\|\boldsymbol{r}[n,m] - \boldsymbol{s}[n] \right\|}
e^{-\frac{\xi(f)}{2}{\left\| {\boldsymbol{r}[n,m] - \boldsymbol{s}[n]} \right\|}},\forall n,m.
\end{equation}
The THz path loss is represented by the exponential term, where $\xi(f)$ is the molecular absorption coefficient, influenced by frequency $f$ and water vapor concentration~\cite{9978646}. For simplicity, we denote it as $\xi$, fixing the $f$. The reference power gain $\fgeeszett_{0} = c/{4 \pi f}$, with $c$ as the speed of light~\cite{9978646}. The channel power gain $h_{sd}[n]$ between the source and destination follows a similar structure as in Eq.~\eqref{h11} and Eq.~\eqref{h12}~\cite{9678373}.
\subsection{Episode One: Direct Transmission and EH with RHS}
In this episode, the source sends information to both the miniature UAV and the destination node using power-domain NOMA.
The UAV, equipped with an RHS, acts as an EH user in this episode.
The radio frequency (RF) source transmitted signal is: 
\begin{equation}\label{transmitted}
s[n] =\sqrt{\fgelb_{1}[n]} s_{1}[n] + \sqrt{\fgelb_{2}[n]} s_{2}[n],\forall n,
\end{equation}
where ${s_1}[n]$ and ${s_2}[n]$ represent the symbols transmitted in each time slot, modeled as independent circularly symmetric complex Gaussian (CSCG) variables with zero mean and unit variance.
Furthermore, $\sqrt{\fgelb_{1}[n]}$ and $\sqrt{\fgelb_{2}[n]}$ correspond to the power allocation coefficients for NOMA in the $n$-th time slot, subject to the following constraints:
\begin{subequations}
\begin{align}
&{\fgelb_{1}[n]}+ {\fgelb_{2}[n]}\leq \fged_{\rm{peak}},\forall n,\label{power1}\\
&\frac{1}{N}\sum\limits_{n = 1}^N {\fgelb_{1}[n]}+ {\fgelb_{2}[n]}\leq \fged_{\max},
\label{power2}
\end{align}
\end{subequations}
where $\fged_{\rm{peak}}$ is the maximum power the source can transmit in any time slot, and $\fged_{\max}$ is the total power constraint across all time slots.
The signal received by the information decoding (ID) antenna and the absorptive EH RHS elements on the miniature UAV from the RF source can be expressed as: 
\begin{align}
y_{\text{ID}}^{(1)}[n] &= 
\sqrt {\varsigma_1[n]}  
h_{su}[n]s[n] + z_1^{(1)}[n],\forall n,
\label{recieved.UAV}\\
y_{\text{EH}}^{(1)}[n] &= 
\sqrt {\varsigma_2[n]}  
\sum\limits_{m=1}^M
g_{sr}[n,m]\textdotbreve{a}[m]e^{j\omega[m]}s[n] 
\nonumber\\ 
&+ z_2^{(1)}[n],\forall n,
\label{recieved.UAV.EH}
\end{align}
where $z_1^{(1)}[n] \sim \mathcal{N}(0,\,\epsilon_{1}^{2})\,$ and ${z^{(1)}_2}[n] \sim \mathcal{N}(0,\,\epsilon_{2}^{2})\,$ 
represent the CSCG noise at the UAV’s ID antenna and the EH RHS, respectively.
Besides, $0<\varsigma_1[n],\varsigma_2[n]<1$ are the received ID and EH power factors. 
The parameters $\textdotbreve{a}[m]$ and $\omega[m]$ denote the absorption coefficient and phase shift applied by the $m$-th element of the RHS.
\begin{remark} 
\textcolor{black}{The absorption coefficients on the RHS are assumed to be uniform, i.e., $\textdotbreve{a}[m] = \textdotbreve{a}_M, \forall m$. Nonlinearity and hardware impairments are not considered in this analysis. Furthermore, no phase shift optimization is performed at the RHS; instead, all phase shifts are uniformly set as $\omega[m] = \omega_M, \forall m$. This is because the RHS in our system is employed exclusively for energy harvesting, where the harvested power depends primarily on the absorption coefficient and the impinging RF power density rather than on coherent phase alignment. In contrast, phase optimization is most relevant when RHS elements are used for beamforming or reflection control in the communication link, which falls outside the scope of this work.} 
\end{remark}

The miniature UAV utilizes successive interference cancellation (SIC) to decode the incoming signals. Specifically, it first decodes the destination node’s data and then subtracts it from the received signal to retrieve its own data. The signal-to-interference-plus-noise ratio (SINR) at the UAV for detecting $s_{2}[n]$ is given by: 
\begin{equation}
\label{SINR1.UAV}
\fgef_{d \leftarrow  u}^{(1)}[n] = 
\frac{{\fgelb_{2}[n]}|h_{su}[n]|^2}
{{\fgelb_{1}[n]}|h_{su}[n]|^2 + \epsilon_{1}^{2}[n]/\varsigma_1[n]},
\forall n.
\end{equation}
Next, the SINR for decoding the miniature UAV's own data is expressed as:
\begin{equation}\label{SINR2.UAV}
\fgef _r^{(1)}[n] = \frac{\varsigma_1[n]{\fgelb_{1}[n]}|h_{su}[n]|^2}{\epsilon_{1}^{2}[n]},\forall n.
\end{equation}
Based on Eq.~\eqref{recieved.UAV} and Eq.~\eqref{recieved.UAV.EH}, the RF power harvested by the EH RHS of the miniature UAV, neglecting the noise power, can be expressed as~\cite{10348506}:
\begin{equation}\label{E.UAV}
\mathcal{E}[n] = \eta  \textdotbreve{a}_M e^{2j\omega_M} \mathcal{T} [n] 
\varsigma_2[n]
{\left|\sum\limits_{m=1}^Mg_{sr}[n,m]\right|^2},
\forall n,
\end{equation}
where $\eta \in (0,1]$ is the energy conversion efficiency, and $\mathcal{T}[n]$ represents the transmission time fraction for the first episode within the $n$-th time slot, assuming equal transmission durations for both episodes. 
Therefore, the UAV's transmit power in the second episode, empowered by the EH RHS, can be written as:
\begin{equation}\label{Pt.UAV}
{\fged_{\rm{RHS}}}[n] = \frac{{\mathcal{E}[n]}}{{1-\mathcal{T} [n]}},\forall n.
\end{equation}
The received signal at the destination is given by:
\begin{equation}\label{recieved from s.d}
y_d^{(1)}[n] =
h_{sd}[n]s[n]+ {\varrho_1}^{(1)}[n],\forall n,
\end{equation}
where ${\varrho_1}^{(1)}[n] \sim \mathcal{N}(0,\,\varepsilon_{1}^2[n])$ is the received noise at the destination node during the first episode.
The SINR at the destination then becomes: 
\begin{equation}\label{SINR phase 1.d}
\fgef _d^{(1)}[n] =  \frac{ {\fgelb_{2}[n]\;|h_{sd}[n]|^2}}{ {\fgelb_{1}[n]\;|h_{sd}[n]|^2} + \varepsilon_{1}^2[n]},\forall n.
\end{equation}
\subsection{Episode Two: Cooperative Transmission}
In this episode, the UAV utilizes the power harvested by the RHS, Eq.~\eqref{Pt.UAV}, to relay the destination node’s data. 
Consequently, the signal received at the destination node is:
\begin{equation}\label{recieved from r.d}
y_d^{(2)}[n] = \sqrt {{\fged_{\rm{RHS}}}[n]}  {h_{ud}}[n]{s_2}[n] +  {\varrho_2}^{(2)}[n],\forall n,
\end{equation}
where ${\varrho_2}^{(2)}[n] \sim \mathcal{N}(0,\,\varepsilon_{2}^2[n])$ is the noise at the destination node. The corresponding SINR is given by: 
\begin{equation}\label{SINR phase 2.d}
\fgef _d^{(2)}[n] = 
\frac{\eta  
\textdotbreve{a}^2_M
e^{2j\omega_M}
\varsigma_2[n]
{\left|\sum\limits_{m=1}^M
g_{sr}[n,m]\right|^2}
{|h_{ud}[n]|^2}}
{\varepsilon_{2}^2[n]},\forall n.
\end{equation}
Finally, the destination node applies maximal ratio combining (MRC) to integrate the signals received in both episodes. The overall SINR can be expressed as: 
\begin{equation}\label{SINR phase 3.d}
\begin{split}
&\fgef _d^{\rm{MRC}}[n] = 
\fgef _d^{(1)}[n]+\fgef _d^{(2)}[n],\forall n.
\end{split}
\end{equation}
\subsection{Resource Allocation Problem Formulation}
We begin by defining the network's EE as the ratio of the total sum rate to the total power consumed by the network. 
Mathematically, this is represented as  
$\eta_{EE}[n]=\frac{R_{\text{sum}}[n]}{\fged_{\text{sum}}[n]}$,
where $R_{\text{sum}}[n]=\log_{2}(1+\fgef _{r}^1[n])+\log_{2}(1+\fgef _{d}^{\text{MRC}}[n])$. 
Assuming a constant power consumption for the miniature UAV’s flight, $\fged_c$, the total transmission power of the system can be written as: $\fged_{\text{sum}}[n]={\fgelb_{1}[n]}+{\fgelb_{2}[n]}+\fged_c-\fged_{\rm{RHS}}[n]$.
To maximize the EE by optimizing the NOMA power allocation coefficients and the UAV's trajectory, we formulate the following optimization problem: 
\begin{subequations}
\begin{align}
&~\text{P}_1: \underset{\fgelb_{1}[n],\fgelb_{2}[n],\boldsymbol{u}[n]}{\max} \: \:\sum\nolimits_{n = 1}^N { \eta_{EE}[n]}
\tag{18}\\
&s.t.: ~
\frac{1}{N}\sum\nolimits_{n = 1}^N 
\fged_{\rm{RHS}}[n]  \ge 
\frac{1}{N}\sum\nolimits_{n = 1}^N 
\fged[n] ,  
\label{p1_a}\\
&~\quad\quad \fgef_{d \leftarrow  u}^{(1)}~[n] 
\ge 
{\fgef_{\min}}[n],\:  \forall n,
\label{p1_b}\\
&~\quad\quad 
\fgef_d^{\rm{MRC}}[n] 
\ge 
{\gamma_{\min }}[n],\:  \forall n,
\label{p1_c}\\
&~\quad\quad \fged[n] \geq 0,\: \forall n,
\label{p1_e}\\
&~\quad\quad 
\eqref{1a}-\eqref{1c}, \eqref{power1}, \eqref{power2}.
\nonumber
\end{align} 
\end{subequations}
The constraint Eq.~\eqref{p1_a} ensures that the power harvested by the EH RHS of the miniature UAV over all time slots is greater than the minimum required harvested power $\fged[n] = \frac{\mathcal{E}[n]}{\mathcal{T}[n]}$ (where $\fged[n] = \fged_{\rm{EH}}$, representing the harvested power). 
Constraint Eq.~\eqref{p1_b} guarantees successful decoding of the destination node’s data at the UAV, with the SINR exceeding the threshold ${\fgef_{\min}}[n]$, while Eq.~\eqref{p1_c} enforces that the destination node’s SINR remains above the minimum requirement ${\gamma_{\min}}[n]$, where $\gamma_{\min} \geq \fgef_{\min}$. Finally, Eq.~\eqref{p1_e} ensures that the UAV's transmitted power is feasible and non-negative.
\section{A Two-Step Sequential Approach to Solving the EE Optimization Problem}\label{section_4}
The optimization problem $\text{P}_1$ is NP-hard and non-convex due to the interdependence among the optimization variables. Additionally, the objective function in $\text{P}_1$ is a sum of ratios, which makes traditional Dinkelbach method approaches unsuitable~\cite{10100913}. To address this, we propose a two-step approach that separates the optimization process, allowing each variable to be optimized independently.
\subsection{Step One: Optimizing EH RHS Miniature UAV Trajectory} 
In this step, the trajectory of the miniature UAV is optimized while the NOMA power allocation coefficients remain fixed. The sum rate function remains non-convex due to the coupling of the optimization variables. However, to address this, the non-linear fractional objective function is first transformed into a subtractive form~\cite{Jong2012AnEG}.\\
\indent\textbf{Theorem}~\cite{Jong2012AnEG}:
Let $\boldsymbol{u}^\ast[n]$ be the optimal solution to $\text{P}_1$. Then, given the existence of two vectors,
$\boldsymbol{\alpha} = [\alpha_1^\ast, \ldots, \alpha_N^\ast]^T$ and $\boldsymbol{\beta} = [\beta_1^\ast, \ldots, \beta_N^\ast]^T$, the following optimization problem provides an optimal solution as follows:
\begin{align}
&\text{P}_2: 
\label{max}
&\underset{\boldsymbol{u}[n]}{\max} \: \:\sum\nolimits_{n = 1}^N {\alpha_n^\ast\big[R_\text{sum}[n]-\beta_n^\ast(\fged_\text{sum}[n])\big]}.
\end{align}
Moreover, $\boldsymbol{u}^\ast[n]$ must satisfy the following conditions: 
\begin{align}
&R_\text{sum}^\ast[n]-\beta_n^\ast(\fged_\text{sum}[n])=0,\forall n,\label{conn1}\\
&1-{\alpha_n^\ast(\fged_\text{sum}[n])}=0,\forall n.\label{conn2}
\end{align}
The equivalent subtractive form in Eq.~\eqref{max}, using the additional parameters ${\boldsymbol{\alpha}^\ast,\boldsymbol{\beta}^\ast}$, shares the same optimal solution as $\text{P}_1$ for fixed values of $\fgelb_1[n]$ and $\fgelb_2[n]$. Specifically, Eq.~\eqref{max} can be solved iteratively using a two-layer approach consisting of inner and outer layers. In the inner layer, Eq.~\eqref{max} is solved with fixed values of $\boldsymbol{\alpha}$ and $\boldsymbol{\beta}$. Then, Eq.~\eqref{conn1} and Eq.\ref{conn2} are updated in the outer layer to find the optimal $\{\boldsymbol{\alpha}^\ast,\boldsymbol{\beta}^\ast\}$.
\textcolor{black}{
\begin{proposition}\label{prp_1}
Problems $\text{P}_1$ and $\text{P}_2$ are equivalent, since for the optimal solution $\boldsymbol{u}^\star[n]$ there exist vectors $\boldsymbol{\alpha}^\star$ and $\boldsymbol{\beta}^\star$ satisfying the conditions in \eqref{conn1}–\eqref{conn2}, which ensure that both problems attain the same optimal value and solution. 
\end{proposition}
\begin{proof}
See Appendix~\ref{Appen_A}.
\QED
\end{proof}
}

\subsubsection{\textbf{Inner-layer}}
Here, we optimize the trajectory based on the optimal NOMA power allocation coefficients as follows
\begin{subequations}
\label{p3}
\begin{align}
& \text{P}_3: 
\underset{\boldsymbol{u}[n]}{\max} 
\: \:
\sum\nolimits_{n = 1}^N 
\alpha_n^\ast\big[R_\text{sum}[n]-\beta_n^\ast(\fged_\text{sum}[n])\big]
\tag{22}\\
&s.t.:  
\sum\limits_{n = 1}^N
\frac{\fgec_1 e^{-\xi
\left(\left\| {\boldsymbol{u}[n] -\boldsymbol{s}[n]} \right\|\right)
}}
{\left\| {\boldsymbol{u}[n] - \boldsymbol{s}[n]} \right\|^2}  
\ge
\sum_{n = 1}^N {\fged[n]} ,  \\
&\frac{\fgelb_2[n]}{\fgelb_1[n] + \fgec_2{{\left\| {\boldsymbol{u}[n] - \mathbf{s}[n]} \right\|}^2
e^{\xi
\left(\left\| {\boldsymbol{u}[n] - \boldsymbol{s}[n]} \right\|\right)}}
}
\ge {\fgef _{\min }}[n],\:  \forall n,\\
&
\frac{\fgec_1
\fgeeszett_{0}^2}{\varepsilon_{2}^2[n]}
\cdot
\frac{e^{-\xi\left(
{\left\|\boldsymbol{u}[n] - \boldsymbol{s}[n]\right\|+{\left\| {\boldsymbol{u}[n] - \boldsymbol{d}[n]} \right\|}}
\right)}}
{\left\| {\boldsymbol{u}[n] - \boldsymbol{s}[n]} \right\|^2\left\| {\boldsymbol{u}[n] - \boldsymbol{d}[n]} \right\|^2}\\
&\quad \quad   
+\frac{ {\fgelb_{2}[n]\;|h_{sd}[n]|^2}}{ {\fgelb_{1}[n]\;|h_{sd}[n]|^2} + \varepsilon_{1}^2[n]}
\ge {\gamma_{\min }}[n],\:  \forall n,
\nonumber
\\
&\quad \quad 
\eqref{1a}-\eqref{1c},\eqref{p1_e}, 
\nonumber
\end{align}
\end{subequations}
where  
$\fgec_1=\eta
M \textdotbreve{a}^2_M  e^{2j\omega_M}
\varsigma_2[n]
\fgeeszett_{0}^2$
and
$\fgec_2=\frac{\epsilon_{1}^{2}[n]}{\varsigma_1[n]\fgeeszett_{0}^2}$. 
The optimization problem $\text{P}_3$ remains non-convex. Therefore, $\text{P}_3$ is reformulated into an equivalent form by introducing slack optimization variables, $(a[n],b[n],c[n],d[n])$, as follows:
\begin{subequations}
\begin{align}
& \text{P}_4: 
\underset{\boldsymbol{u}[n],a[n],b[n],c[n],d[n]}{\max} \: \:\sum\nolimits_{n = 1}^N 
{\alpha_n^\ast\big[R_\text{sum}[n]-\beta_n^\ast(\fged_\text{sum}[n])\big]}
\tag{23}\\
&s.t.: \quad \sum\nolimits_{n = 1}^N 
\frac{\fgec_2}{e^{c[n]}}  
\ge 
\sum\nolimits_{n = 1}^N {\fged[n]},\\  
&\quad\quad \frac{{{\fgelb_2}[n]}}{{{\fgelb_1}[n] + \fgec_2e^{c[n]}}}\ge {\fgef _{\min }}[n],\:  \forall n,
\\
&\quad\quad
\frac{\fgelb_{2}[n]\;|h_{sd}[n]|^2}
{\fgelb_{1}[n]\;|h_{sd}[n]|^2+ 
\varepsilon_{1}^2[n]}+\dots
\nonumber\\
&\quad\quad\quad\quad\quad\quad
+\frac{\fgec_1\fgeeszett_{0}^2}{\varepsilon_{2}^2[n]{e^{c[n]+d[n]}}}\ge {\gamma_{\min}}[n],\forall n,
\\
& \quad\quad
a[n]
\leq 
\frac{\left\| {\boldsymbol{u}[n] - \boldsymbol{s}[n]} \right\|^2}
{e^{-\xi\left\| {\boldsymbol{u}[n] - \boldsymbol{s}[n]} \right|}}, \forall n,
\\
& \quad\quad
b[n]
\leq
\frac{\left\|{\boldsymbol{u}[n] - \boldsymbol{d}[n]} \right\|^2}
{e^{-\xi\left\|{\boldsymbol{u}[n] - \boldsymbol{d}[n]} \right\|}}, \forall n,
\\ 
& \quad\quad
a[n]\leq e^{c[n]},  \forall n,
\\
& \quad\quad b[n]\leq e^{d[n]}, \forall n,
\\&~\quad\quad \nonumber
\eqref{1a}-\eqref{1c},\eqref{p1_e}, 
\end{align}
\end{subequations}
where 
\begin{align}
R_\text{sum}[n]&=\log_{2}
\left(1+
\frac{\fgec_2\fgelb_1[n]}{e^{c[n]}}\right)
\\
&
+\log_{2}
\left(
1+\fgef_d^{(1)}[n]+ 
\left(
\frac{\fgec_1\fgeeszett_{0}^2}{\varepsilon_{2}^2[n]}\cdot\frac{1}{e^{c[n]+d[n]}}\right)
\right).\nonumber
\end{align}
Using these transformations, the main objective function and constraints become convex but still intractable. Therefore, successive convex approximation (SCA) using first-order Taylor expansions is applied to approximate $\text{P}_4$ as convex functions. The first-order lower bounds are given by: 
\begin{align}
&e^{c[n]}\geq 
e^{c^{(k)}[n]}
(1+{c[n]}-{c^{(k)}[n]})
\buildrel \Delta \over = 
\tilde{e}^{c[n]},\forall n,\\
&e^{d[n]}\geq 
e^{d^{(k)}[n]}
(1+\ {d[n]}-{d^{(k)}[n]})
\buildrel \Delta \over = 
\tilde{e}^{d[n]},\forall n,
\end{align} 
\begin{figure*}[t!]
\centering
\begin{align}  
L_{\mathfrak{z}}(\fgelb_{1}[n],\fgelb_{2}[n],\boldsymbol{\mathfrak{a}},\boldsymbol{\mathfrak{b}},\boldsymbol{\Upsilon},\boldsymbol{\mathfrak{c}},\boldsymbol{\mathfrak{d}})
&=
\sum\limits_{n = 1}^N \varpi[n]{\fged_\text{sum}^2[n]} + 
\frac{1}{2\mathfrak{z}}
\Bigg[
\bigg(
\bigg[
\sum\limits_{n = 1}^N\mathfrak{a}_n
+\mathfrak{z} \Big(
\frac{\epsilon_{1}^{2}[n]}{\varsigma_1[n]}
-
\frac{{\fgelb_{2}[n]}|h_{su}[n]|^2}{{\fgef_{\min}}[n]}+{\fgelb_{1}[n]}|h_{su}[n]|^2 
\Big)\bigg]^+
\bigg)^2\nonumber
\\
&+
\sum\limits_{n = 1}^N 
 \frac{1}{4\varpi[n]\hat{R}_\text{sum}^2[n]}
+
\bigg(
\bigg[
\sum\limits_{n = 1}^N \mathfrak{b}_n
+
\mathfrak{z}(
\varepsilon_{1}^2[n]\chi[n]
-
{\fgelb_{2}[n]}|h_{sd}[n]|^2
+
{\fgelb_{1}[n]}|h_{sd}[n]|^2\chi[n]   
)\bigg]^{+}\bigg)^{2}
\nonumber\\\nonumber
&+
\bigg(\bigg[
\sum\limits_{n = 1}^N \Upsilon_n
+
\mathfrak{z}({ \fgelb_{1}[n] } + { \fgelb_{2}[n] }- \fged_{\rm{peak}})\bigg]^{+}\bigg)^{2}\nonumber
\end{align}
\begin{align}\label{lagrange_aug}
~&+
\bigg(
\bigg[ 
\mathfrak{c}_n
+\mathfrak{z}(\frac{1}{N}\sum\limits_{n = 1}^N { \fgelb_{1}[n] } 
+\fgelb _2[n] -\fged_{\max})\bigg]^{+}\bigg)^{2}
\nonumber\\
&+
\bigg(\bigg[ \sum\limits_{n = 1}^N\mathfrak{d}_n - \mathfrak{z} \fged[n] \bigg]^{+}\bigg)^{2}
-\bigg(
\sum\limits_{n = 1}^N 
\mathfrak{a}_n^2+
\mathfrak{b}_n^2 +
\Upsilon_n^2+
\mathfrak{c}_n^2 +
\mathfrak{d}_n^2\bigg)
\Bigg],\tag{42}
\end{align}
\medskip
\hrule
\end{figure*}
\begin{align}
\frac{\left\| {\boldsymbol{u}[n] - \boldsymbol{s}[n]} \right\|^2}
{e^{-\xi\left\| {\boldsymbol{u}[n] - \boldsymbol{s}[n]} \right\|}}
&\geq 
\frac{\left\| \boldsymbol{u}^{(k)}[n] - \boldsymbol{s}[n] \right\|^2}{e^{-\xi\left\| \boldsymbol{u}^{(k)}[n] - \boldsymbol{s}[n] \right\|}}
\nonumber
\\&+
(2+\xi||\boldsymbol{u}^{(k)}[n] - \boldsymbol{s}[n]||)
\cdot
\end{align} 
\begin{align}\hspace{-2mm}
\frac{(\boldsymbol{u}^{(k)}[n]-\boldsymbol{s}[n])^T
(\boldsymbol{u}[n]-{{\boldsymbol{u}}^{(k)}}[n])}{e^{-\xi\left\| \boldsymbol{u}^{(k)}[n] - \boldsymbol{s}[n] \right\|}}
\buildrel 
\Delta \over = 
\frac{\left\| \tilde{\boldsymbol{u}}[n] - \boldsymbol{s}[n] \right\|^2}{e^{-\xi\left\|\tilde{\boldsymbol{u}}[n]-\boldsymbol{s}[n]\right\|}},  \forall n,
\label{taylor1} 
\end{align} 
\begin{align}
\frac{\left\| {\boldsymbol{u}[n] - \boldsymbol{d}[n]} \right\|^2}{e^{-\xi \left\| {\boldsymbol{u}[n] - \boldsymbol{d}[n]} \right\|^2}}
&\geq 
\frac{\left\| \boldsymbol{u}^{(k)}[n] - \boldsymbol{d}[n] \right\|^2}{e^{-\xi \left\| \boldsymbol{u}^{(k)}[n] - \boldsymbol{d}[n] \right\|}}
\nonumber
\\&+
(2+\xi || \boldsymbol{u}^{(k)}[n] - \boldsymbol{d}[n] ||)
\cdot
\end{align} 
\begin{align}\hspace{-3mm}
\frac{(\boldsymbol{u}^{(k)}[n]-\boldsymbol{d}[n])^T(\boldsymbol{u}[n]-{{\boldsymbol{u}}^{(k)}}[n])}{e^{-\xi \left\| \boldsymbol{u}^{(k)}[n] - \boldsymbol{d}[n] \right\|}}
\buildrel \Delta \over = 
\frac{\left\| {\tilde{\boldsymbol{u}}[n] - \boldsymbol{d}[n]} \right\|^2}{e^{-\xi \left\| {\tilde{\boldsymbol{u}}[n] - \boldsymbol{d}[n]} \right\|^2}},  \forall n,
\label{taylor2}
\end{align}    
where 
$e^{c^{(k)}[n]}$ 
and 
$e^{d^{(k)}[n]}$
represent the Taylor expansion points at iteration $k$.
With this transformation, $\text{P}_4$'s approximation becomes:
\begin{subequations}
\label{p5}
\begin{align}
& \text{P}_5: 
\underset{\boldsymbol{u}[n],a[n],b[n],c[n],d[n]}{\max} 
\sum\nolimits_{n = 1}^N \alpha_n^\ast\big[\tilde{R}_\text{sum}[n]-\beta_n^\ast(\fged_\text{sum}[n])\big]
\tag{31}\\
&s.t.: \quad \sum\nolimits_{n = 1}^N 
\frac{\fgec_1}{\tilde{e}^{c[n]}}
\ge \sum\nolimits_{n = 1}^N {\fged[n]} ,  \\
&\quad\quad \frac{{{\fgelb _2}[n]}}{{{\fgelb _1}[n] + \fgec_2\tilde{e}^{c[n]}}}\ge {\fgef _{\min }}[n],\:  \forall n\\
&\quad\quad
\frac{ {\fgelb_{2}[n]\;|h_{sd}[n]|^2}}{ {\fgelb_{1}[n]|h_{sd}[n]|^2} + \varepsilon_{1}^2[n]}+
\frac{\fgec_1\fgeeszett_{0}^2}
{\varepsilon_{2}^2[n]{\tilde{e}^{c[n]+d[n]}}}
\ge 
{\gamma_{\min}}[n],\forall n,\\
&\quad\quad
a[n]
\leq 
\frac{\left\| {\tilde{\boldsymbol{u}}[n] - \boldsymbol{s}[n]} \right\|^2}{e^{-\xi\left\| {\tilde{\boldsymbol{u}}[n] - \boldsymbol{s}[n]} \right\|}},\forall n,\\
&\quad\quad
b[n]
\leq 
\frac{\left\| {\tilde{\boldsymbol{u}}[n] - \boldsymbol{d}[n]} \right\|^2}{e^{-\xi\left\| {\tilde{\boldsymbol{u}}[n] - \boldsymbol{d}[n]} \right\|}},\forall n,
\\
&\quad  \quad 
a[n]
\leq \tilde{e}^{c[n]}, \forall n,
\\
&\quad\quad 
b[n]
\leq 
\tilde{e}^{d[n]},\forall n,
\\&~\quad\quad \nonumber
\eqref{1a}-\eqref{1c},\eqref{p1_e}, 
\end{align}
\end{subequations}
where $\tilde{R}_\text{sum}[n]={R}_\text{sum}[n]|_{{e}^{c[n]}=\tilde{e}^{c[n]},{e}^{d[n]}=\tilde{e}^{d[n]}}$. 
Optimization solvers can be employed to find a solution for $\text{P}_5$\cite{10100913}. 
\subsubsection{\textbf{Outer-layer}}
The damped Newton method is applied to find the optimal values for
$\{\boldsymbol{\alpha},\boldsymbol{\beta}\}$. 
Let $\theta_n(\beta_n)=R_\text{sum}^\ast[n]-\beta_n^\ast(\fged_\text{sum}[n])$ and 
$\theta_{N+j}(\alpha_j)=1-
{\alpha_j^\ast(\fged_\text{sum}{[j]})}$, $j\in \{1,...,N\}$. 
As shown in \cite{10059879}, 
the solution 
$\{\boldsymbol{\alpha}^\ast,\boldsymbol{\beta}^\ast\}$ is optimal if and only if $\theta(\boldsymbol{\alpha},\boldsymbol{\beta})=[\theta_1,\theta_2,...,\theta_{2N}]^T=0$. 
The updated values of $\boldsymbol{\alpha}^{i+1}$ and $\boldsymbol{\beta}^{i+1}$ can be computed by: 
\begin{align}
&\boldsymbol{\alpha}^{i+1}=\boldsymbol{\alpha}^{i}+\vartheta ^i\boldsymbol{\mu}^i_{N+1:2N},\label{alphacon1}\\
&\boldsymbol{\beta}^{i+1}=\boldsymbol{\beta}^{i}+\vartheta ^i\boldsymbol{\mu}^i_{1:N},\label{alphacon2}
\end{align}
where $\boldsymbol{\mu}=[\acute{\theta}(\boldsymbol{\alpha},\boldsymbol{\beta})]^{-1}\theta(\boldsymbol{\alpha},\boldsymbol{\beta})$ with $\acute{\theta}(\boldsymbol{\alpha},\boldsymbol{\beta})$ being the Jacobian matrix of ${\theta}(\boldsymbol{\alpha},\boldsymbol{\beta})$, and $\vartheta ^i$ is the largest value of $\Uppi^m$ at iteration $i$ satisfying:
\begin{equation}\label{con}
\|\theta(\boldsymbol{\alpha}^{i}+\Uppi^m\boldsymbol{\mu}^i_{N+1:2N},\boldsymbol{\beta}^{i}+\Uppi^m\boldsymbol{\mu}^i_{1:N})\|\leq(1-\wp\Uppi^m)\|\theta(\boldsymbol{\alpha},\boldsymbol{\beta})\|,
\end{equation}
where $m\in \{1,2,...\}$, $\Uppi^m \in (0,1)$, and $\wp \in (0,1)$. 
\subsection{Step two: Optimizing NOMA Power Allocation Coefficients}
Consider the following sum-fraction optimization problem:
\begin{equation}
\mathop 
{\min}
\limits_{\boldsymbol{\Omega} \in C}  
\sum\limits_{j = 1}^J 
\frac{\mathcal{A}_j(\boldsymbol{\Omega})}{\mathcal{B}_j(\boldsymbol{\Omega})},
\label{stage_2_1}
\end{equation}
where $J$ represents the total number of fractional terms, and $\boldsymbol{\Omega}$ is the vector of optimization variables within the feasible domain $C$.
It can be shown Eq.~\eqref{stage_2_1} is equivalent to:
\begin{equation}
\mathop 
{\min }
\limits_{\boldsymbol{\Omega} \in C,\varpi_j > 0}  
\sum\limits_{j = 1}^J 
\varpi_j
\mathcal{A}_j^2(\boldsymbol{\Omega}) + 
\sum\limits_{j = 1}^J 
\frac{1}{4\varpi_j} 
\frac{1}{\mathcal{B}_j^2(\boldsymbol{\Omega})}.
\label{stage_2_2}
\end{equation}
The solution to both Eq.~\eqref{stage_2_1} and Eq.~\eqref{stage_2_2} is identical.
It is worth noting that if $\mathcal{B}_j(\boldsymbol{\Omega})$ is concave and $\mathcal{A}_j(\boldsymbol{\Omega})$ is convex, then problem in Eq.~\eqref{stage_2_2} becomes a convex quadratic problem for the given $\varpi_j$.
Building on this, the convex problem in Eq.~\eqref{stage_2_2} is solved for a given $\varpi_j = {1}/{{2{\mathcal{B}_j}(\boldsymbol{\Omega}){\mathcal{A}_j}(\boldsymbol{\Omega})}}$, and the value of $\varpi_j$ is updated in the next iteration.
Thus, with a fixed UAV trajectory, problem $\text{P}_1$ can be rewritten in the following equivalent form:
\begin{align}
& \text{P}_6: 
\underset{\fgelb_{1}[n],\fgelb_{2}[n],\varpi[n] > 0}{\min}
\:\sum\limits_{n = 1}^N  
\varpi[n]{\fged_\text{sum}^2[n]} 
\nonumber
\\&\quad\quad\quad\quad\quad\quad\quad\quad ~+
\sum\limits_{n = 1}^N 
\frac{1}{4\varpi[n]} \frac{1}{R_\text{sum}^2[n]}\\
&s.t.: \frac{{\fgelb_{2}[n]}|h_{su}[n]|^2}{{\fgef_{\min }}[n]}-{\fgelb_{1}[n]}|h_{su}[n]|^2\geq \frac{\epsilon_{1}^{2}[n]}{\varsigma_1[n]},  \forall n,
\label{p6_a}
\tag{36a}
\end{align}
\begin{subequations}
\begin{align}
&\hspace{-2mm}
{\fgelb_{2}[n]}|h_{sd}[n]|^2-
{\fgelb_{1}[n]}|h_{sd}[n]|^2\chi[n]
\geq 
\varepsilon_{1}^2[n]\chi[n],  \forall n,
\label{p6_b}
\tag{36b}
\\
&\quad\quad  \eqref{power1}, \eqref{power2},\eqref{p1_e},\nonumber
\end{align}    
\end{subequations}
where 
$\chi[n]=\gamma_{\min}[n]-
\frac{\fgec_1
{\left|\sum\limits_{m=1}^Mg_{sr}[n,m]\right|^2}{|h_{ud}[n]|^2}}{M\fgeeszett^2_0\varepsilon_{2}^2[n]}$ 
and
$\varpi[n]=\frac{1}{2\fged_\text{sum}^2[n]R_\text{sum}^2[n]}$. 
It is evident that all constraints are linear and convex. However, the objective function remains non-convex due to the non-concave nature of the sum rate function. To address this, we apply the result from the following corollary~\cite{10100913}.
\begin{corollary}\label{corollary1}
Let $\mathcal{F}$ be a monotonically decreasing function of the ratio $\frac{\mathcal{C}_j(\boldsymbol{\mho})}{\mathcal{D}_j(\boldsymbol{\mho})}$.
The optimization problem
\begin{equation}
\tag{37}
\min_{\boldsymbol{\mho} \in C} \sum_{j=1}^{J} \mathcal{F}_j \left( \frac{\mathcal{C}_j(\boldsymbol{\mho})}{\mathcal{D}_j(\boldsymbol{\mho})} \right),
\end{equation}
is equivalent to:
\begin{equation}
\tag{38}
\min_{\boldsymbol{\mho} \in C, \lambda_j} \sum_{j=1}^{J} \mathcal{F}_j \left( 2\lambda_j \sqrt{\mathcal{C}_j(\boldsymbol{\mho})} - \lambda_j^2 \mathcal{D}_j(\boldsymbol{\mho}) \right),
\end{equation}
where $\lambda_j$ is updated iteratively as:
$\lambda_j = \frac{\sqrt{\mathcal{C}_j(\boldsymbol{\mho})}}{\mathcal{D}_j(\boldsymbol{\mho})}$.
\end{corollary}

By applying the result from Corollary \ref{corollary1}, the second term in the objective function of $\text{P}_6$ can be rewritten as:
\begin{equation}
\tag{39}
\underset{\fgelb_{1}[n],\fgelb_{2}[n],\lambda[n]}{\min} \:\:  \sum\limits_{n = 1}^N  {\frac{1}{{4{\varpi[n]}}}} \frac{1}{\hat{R}_\text{sum}^2[n]},
\end{equation}
where 
\begin{align}
{\hat{R}_\text{sum}[n]}\nonumber
&=\log_{2}( 1+\fgef_r^1[n])
\\&\nonumber+
\log_{2}
\bigg(1+\fgef_d^2[n]
+2\lambda[n]\sqrt{ {\fgelb_{2}[n]\;|h_{sd}[n]|^2}}
\\&~~~~~~~~~~-
\lambda^2[n](\fgelb_{1}[n]\;|h_{sd}[n]|^2 
+ \varepsilon_{1}^2[n])\bigg),
\tag{40}
\end{align}
with
$\lambda[n]=
\frac{\sqrt{\fgelb_{2}[n]\;|h_{sd}[n]|^2}}{\fgelb_{1}[n]\;|h_{sd}[n]|^2 
+ \varepsilon_{1}^2[n]}$.
Here, $\hat{R}_\text{sum}[n]$ becomes biconcave in terms of both the power allocation coefficients and $\lambda[n]$.
Consequently, the multi-convex optimization problem is formulated as: 
\begin{subequations}
\begin{align}
& \text{P}_7: \underset{\gimel[n],\varpi[n]}{\min} \: \:\sum\limits_{n = 1}^N { {{\varpi[n]}{\fged_\text{sum}^2[n]} + \sum\limits_{n = 1}^N {\frac{1}{{4{\varpi[n]}}}} \frac{1}{\hat{R}_\text{sum}^2[n]}} }
\tag{41}
\\
&s.t.: \quad \eqref{power1}, \eqref{power2}, \eqref{p1_e}, \eqref{p6_a}, \eqref{p6_b},
\nonumber
\end{align}
\end{subequations}
where $\gimel[n]=[\fgelb_{1}[n],\fgelb_{2}[n]]\in \mathbb{R}^{2\times 1}$. 
Note that $\fged_\text{sum}[n]$ depends on the power allocation coefficients, and each coefficient is subject to its respective constraints. 
Thus, $\fged_\text{sum}[n]$ and $\hat{R}_\text{sum}[n]$ are decoupled to enable the distributed optimization of $\fged_\text{sum}[n]$. 
To solve this, the augmented Lagrangian method (ALM) is applied, as defined in Eq.~\eqref{lagrange_aug}, where a penalty term is introduced in the Lagrange function of $\text{P}_7$, yielding a sub-optimal solution.
In Eq.~\eqref{lagrange_aug}, $\mathfrak{z}$ represents the penalty factor, while ${\boldsymbol{\mathfrak{a}}, \boldsymbol{\mathfrak{b}}, \boldsymbol{\Upsilon}, \boldsymbol{\mathfrak{c}}, \boldsymbol{\mathfrak{d}}}$ are the Lagrange multipliers.
These components work together to steer the optimization process towards solutions that are feasible within the problem's constraints while penalizing deviations from these constraints to maintain a strict adherence to them. Each iteration of the optimization, denoted as $l$, yields a solution $B^{(l)}$, progressively refining the approach towards an optimal or sub-optimal solution.
Finally, our proposed efficient low complexity sub-optimal algorithm is sketched in~\textbf{Algorithm~\ref{chap8_alg1}}.

\textcolor{black}{
\begin{proposition}\label{prp_2}
The objective function value of $\text{P}_1$ would be improved via the iterative algorithm.
\end{proposition}
\begin{proof}
See Appendix~\ref{Appen_B}.
\QED
\end{proof}
}

\begin{algorithm}[t]\textcolor{black}{
\caption{\textcolor{black}{Iterative Resource Allocation Algorithm for EE Maximization of Miniature UAV-Aided Cooperative THz Networks with Reconfigurable Energy Harvesting Holographic Surfaces}}
\begin{algorithmic}[1]
\renewcommand{\algorithmicrequire}{\textbf{Input:}}
\renewcommand{\algorithmicensure}{\textbf{Output:}}
\REQUIRE Set iteration indices $i=0,k=0,l=0$,\\
\ \ \ \:\:Set the maximum convergence iteration index $I_{\max}$,  \\
\ \ \ \:\:Set the tolerance to $\epsilon_1=\epsilon_2=10^{-3}$,\\
\ \ \ \:\:Initialize 
$\boldsymbol{\alpha}$, 
$\boldsymbol{\beta}$, \\
\ \ \ \:\:Initialize $a^{(k)}[n], b^{(k)}[n], c^{(k)}[n],d^{(k)}[n]$,
$\boldsymbol{u}^{(k)}$,
$\gimel$,\\
\ \ \ \:\:Set Lagrange multipliers 
${\boldsymbol{\mathfrak{a}}^{l}, \boldsymbol{\mathfrak{b}}^{l}, \boldsymbol{\Upsilon}^{l}, \boldsymbol{\mathfrak{c}}^{l}, \boldsymbol{\mathfrak{d}}^{l}}$, \\
\ \ \ \:\:Set the penalty factor $\mathfrak{z}^{l}.$
\STATE \textbf{repeat}\\
\STATE \quad \textbf{while} $|A^{(k)}-A^{(k-1)}|\geq \epsilon_1$ do \\
\STATE \quad \quad Given $\gimel[n]$, solve $\text{P}_5$ to obtain $\boldsymbol{u}^{(k)}[n].$
\STATE \quad \quad Update $b^{(k)}=\ln(c^{(k)}[n])$, $a^{(k)}=\ln(d^{(k)}[n])$ ,\\
\ \ \ \ \ \:according to \eqref{taylor1}  and \eqref{taylor2}.
\STATE \quad \quad Set $k=k+1$.
\STATE   \quad \textbf{end while} 
\STATE \quad \textbf{if} \eqref{con} is satisfied \textbf{then}~\textbf{return}~$\boldsymbol{u}^\ast[n]$.
\STATE \quad \textbf{else}~Update $\boldsymbol{\alpha}$ and $\boldsymbol{\beta}$ according to (\ref{alphacon1}) and (\ref{alphacon2}).
\STATE \quad Set $i=i+1$.
\STATE   \textbf{until} \eqref{conn1} and \eqref{conn2} are satisfied or $i=I_{{\max}}$.
\STATE \textbf{while} $|B^{(l)}-B^{(l-1)}|\geq \epsilon_2$ do \\
\STATE \quad Given $\boldsymbol{u}[n]$, solve $\text{P}_7$ to obtain $\gimel[n]$.
\STATE \quad Update the Lagrange multipliers
$\boldsymbol{\mathfrak{a}}^{l+1}_n$, $\boldsymbol{\mathfrak{b}}^{l+1}_n$, $\boldsymbol{\Upsilon}^{l+1}_n$, \\
\ \ \ \:$\boldsymbol{\mathfrak{c}}^{l+1}_n$, and $\boldsymbol{\mathfrak{d}}^{l+1}_n$.
\STATE \quad Update the penalty factor $\mathfrak{z}^{l+1}=2\mathfrak{z}^{l}$.
\STATE \quad Set $l=l+1$.
\STATE   \textbf{end while} 
\STATE   \textbf{return} $(\boldsymbol{u}^\ast[n],\gimel^\ast[n])$.
\end{algorithmic}\label{chap8_alg1}}
\end{algorithm}

\textcolor{black}{\section{Complexity Analysis}\label{chap8_sec5}
The overall complexity of the proposed two-stage solution is determined by the complexities of solving two optimization problems: $\text{P}_5$ and $\text{P}_7$, associated with finding the optimal miniature UAV trajectory and the NOMA power coefficients.
$\text{P}_5$ has \( (8N+3) \) constraints and \( 5N \) decision variables. Its complexity, based on the Successive Convex Approximation methodology, is \( O_1 = \mathcal{O}((8N+3)(5N)^3) \).
The complexity of $\text{P}_7$, following the Augmented Lagrangian Method, is \( O_2 = \mathcal{O}(N^2) \).
Hence, the total complexity of the proposed solution is the sum of the individual complexities: \( O_{\rm{total}} = O_1 + O_2  = \mathcal{O}((8N+3)(5N)^3+N^2) \), indicating a polynomial time complexity of degree four~\cite{10100913}.}

\section{Simulation Results and Discussions}\label{chap8_sec6}

\begin{table}[!t]
\caption{Simulation Parameters for EE Maximization of THz-NOMA Networks Empowered by Holographic Surfaces for Miniature UAVs.}
\label{simulation_parameters}
\centering
\begin{tabular}{|l|l|}
\hline
\textbf{Parameter}           & \textbf{Value}                   \\ \hline
Area side length             & $30~\rm{meters}$                      \\
Carrier frequency            & $1.2$ THz                        \\ 
Transmission bandwidth       & $10$ GHz                         \\ 
Absorption coefficient, $\xi(f)$ & 0.005           \\ 
RHS absorption coefficient, $\textdotbreve{a}_M$ & 1           \\ 
Maximum miniature UAV flying speed, $V_{\max}$ & $1~\rm{meter/second}$                     \\
Duration of each time slot, $\Delta_t$          & $0.1~\rm{second}$                     \\ 
Miniature UAV Operation time, $T$         & $45~\rm{second}$                      \\ 
Noise power spectral density & $-174~\rm{dBm/Hz}$               \\ 
Source Node altitude, $H_s$          & $2~\rm{meters}$                       \\ 
Miniature UAV altitude, $H_u$          & $3~\rm{meters}$                       \\ 
Peak power, $\fged_{\rm{peak}}$  & $1$ $\rm{W}$                            \\
Circuit power, $\fged_c$         & $0.52$ $\rm{W}$                         \\ \hline
\end{tabular}
\end{table}

Our simulation setup involves a scenario within a square area, each side being $30~\rm{meters}$, containing one user and a miniature UAV, both randomly placed. To minimize path loss peaks, the carrier frequency is set to $f = 1.2$ THz suggested by~\cite{9838676,9370130} with a transmission bandwidth of $10$ GHz. The model also considers the frequency-dependent absorption coefficient, $\xi(f)$, which accounts for molecular absorption loss due to water vapor~\cite{5995306}. 
All statistical results are derived from aggregating data gathered through an extensive set of simulation trials, which include 1000 random realizations of channel gains. This systematic approach provides an in-depth understanding of the dynamics associated with deploying and operating the miniature UAV under specified environmental conditions, offering essential insights for optimizing UAV-assisted communication networks. A summary of all simulation parameters studied in this paper is presented in \textit{Table~\ref{simulation_parameters}}, as suggested in~\cite{8752399,9978646,10464825}.

\begin{figure}[t]
\centering
\begin{tikzpicture}
\begin{axis}[
    width=0.49\textwidth,
    height=0.40\textwidth,
    xlabel={$\bar{p}_{\rm{sum}}$ [$\rm{W}$]},
    ylabel={$\eta_{EE}$ [$\rm{Mbits/Joule}$]},
    grid=major,
    legend style={
        at={(0.99,0.02)}, 
        anchor=south east, 
        font=\scriptsize,
        legend cell align={left},
        legend columns=2,
    },
    tick label style={font=\small},
    xlabel style={font=\small},
    ylabel style={font=\small},
    xmin=1, xmax=8,
    xtick={1,2,3,4,5,6,7,8},
    ymin=30, ymax=80,
    ytick={30,40,50,60,70,80},
    legend entries={
        \ \ \ \ \ \ \ Proposed,
        \!\!\!\!\!\!\!\!\!\!\!\!\!\!\!Solution,
        Method A,
        Method B,
        Method C,
        Method D,
        Method E,
        Initial
    }
]

\addplot[mark=pentagon, solid, black, line width=0.85pt] coordinates {
    (1, 47.2) (2, 57.0) (3, 62.7) (4, 67.7) (5, 69.9) (6, 71.3) (7, 72.2) (8, 73.6)
};
\addplot[solid, white, line width=0.85pt] coordinates {
    (1, 1)
};
\addplot[mark=o, solid, lightpurple, line width=0.85pt] coordinates {
    (1, 42.4) (2, 53.9) (3, 59.6) (4, 62.9) (5, 65.8) (6, 67.1) (7, 68.4) (8, 69.5)
};
\addplot[mark=10-pointed star, solid, red, line width=0.85pt] coordinates {
    (1, 39.5) (2, 51.3) (3, 56.8) (4, 60.5) (5, 64.0) (6, 65.7) (7, 67.2) (8, 68.5)
};
\addplot[mark=diamond, solid, darkgreen, line width=0.85pt] coordinates {
    (1, 38.5) (2, 49.0) (3, 54.1) (4, 58.2) (5, 62.5) (6, 63.9) (7, 65.8) (8, 66.7)
};
\addplot[mark=square, solid, blue, line width=0.85pt] coordinates {
    (1, 34.2) (2, 44.0) (3, 49.7) (4, 54.7) (5, 56.9) (6, 58.3) (7, 59.2) (8, 60.6)
};
\addplot[mark=triangle, solid, cyan, line width=0.85pt] coordinates {
    (1, 32.5) (2, 42.3) (3, 48.2) (4, 53.2) (5, 55.4) (6, 56.8) (7, 57.7) (8, 58)
};

\addplot[mark=+,densely dashdotted, brown, line width=0.85pt] coordinates {
    (1, 31.2) (2, 41.5) (3, 46.9) (4, 51.7) (5, 53.9) (6, 55.3) (7, 56.2) (8, 56.5)
};

\end{axis}
\end{tikzpicture}
\caption{The impact of average network transmit power, $\Bar{p}_{\text{sum}}$, on the EE of THz-NOMA networks with a miniature UAV empowered by holographic surfaces.}
\label{EE_vs_p_sum}
\end{figure}

To thoroughly assess the performance of our proposed resource allocation algorithm, we conducted a comparative study using the following benchmarks, each selected to highlight different system aspects:
\begin{itemize}
\item \textcolor{black}{Method A: Assesses the algorithm's performance within a fixed NOMA framework, providing a baseline for how the algorithm performs with static power coefficients. }
\item \textcolor{black}{Method B: Compares NOMA and orthogonal multiple access (OMA) to identify which access scheme is more efficient, crucial for understanding the benefits of multi-user communication in this context.}
\item \textcolor{black}{Method C: Tests the algorithm with a fixed UAV flight path, isolating the effects of UAV trajectory optimization and measuring its contribution to overall performance. }
\item \textcolor{black}{Method D: Evaluates a scenario without RHS, focusing on consistent power splitting EH antenna at the UAV antenna. This shows the impact of RHS on EH and EH.}
\item \textcolor{black}{Method E: Implements a fractional programming approach~\cite{9968296} without RHS for comparison, highlighting the benefits of incorporating RHS in our proposed solution.}
\end{itemize}

Fig.~\ref{EE_vs_p_sum} illustrates the EE dynamics as influenced by the average network transmission power, expressed as $\Bar{p}_{\text{sum}}=\fged_{\max}+\fged_{\rm{peak}}+\fged_{c}-\fged_{\rm{EH}}$. In this figure, the 'Initial' curve depicts the EE performance based on an initial, unoptimized (random) configuration of the miniature UAV's flight path. A key finding from our analysis is that our proposed algorithm consistently surpasses various benchmark methods, with its advantage becoming more evident as $\Bar{p}_{\text{sum}}$ increases, resulting in a widening performance gap.
The results demonstrated the effectiveness of our proposed approach, showing improvements of: 30.3\% over Method E, 23.0\% over Method D, 21.2\% over Method C, 18.1\% over Method B, and 7.26\% over Method A. These results strongly affirm the proposed algorithm’s ability to significantly boost EE, proofing its effectiveness within miniature UAV-empowered RHS communication networks.

\begin{figure}[!t]
\centering
\begin{tikzpicture}
\begin{axis}[
    width=0.49\textwidth,
    height=0.40\textwidth,
    xlabel={Miniature UAV Operation time [$\rm{s}$]},
    ylabel={$\eta_{EE}$ [$\rm{Mbits/Joule}$]},
    grid=major,
    legend style={
        at={(0.99,0.02)}, 
        anchor=south east, 
        font=\scriptsize,
        legend cell align={left},
        legend columns=2,
    },
    tick label style={font=\small},
    xlabel style={font=\small},
    ylabel style={font=\small},
    xmin=5, xmax=40,
    xtick={5,10,...,40},
    ymin=20, ymax=50,
    ytick={20,25,...,50},
    legend entries={
        Proposed Solution,
        Method A,
        $~~~~$Method B,
        Method C,
        $~~~~$Method D,
        Initial
    }
]

\addplot[mark=pentagon, solid, black, line width=0.85pt] coordinates {
   (5, 42.7) (10, 45.1) (15, 47.2) (20, 47.8) (25, 47.4) (30, 48.1) (35, 48.1) (40, 48.1) 
};
\addplot[mark=o, solid, lightpurple, line width=0.85pt] coordinates {
    (5, 38.5) (10, 41.2) (15, 43.4) (20, 44.3) (25, 43.8) (30, 44.5) (35, 44.5) (40, 44.5)
};
\addplot[mark=10-pointed star, solid, red, line width=0.85pt] coordinates {
    (5, 37.7) (10, 40.1) (15, 42.2) (20, 42.8) (25, 43.4) (30, 43.1) (35, 43.1) (40, 43.1)
};
\addplot[mark=diamond, solid, darkgreen, line width=0.85pt] coordinates {
    (5, 35.5) (10, 37.9) (15, 40.0) (20, 40.6) (25, 40.2) (30, 40.9) (35, 40.9) (40, 40.9)
};
\addplot[mark=square, solid, blue, line width=0.85pt] coordinates {
   (5, 32.2) (10, 33.6) (15, 35.1) (20, 36.3) (25, 37.5) (30, 38.2) (35, 39.0) (40, 39.6)
};
\addplot[mark=+,densely dashdotted, brown, line width=0.85pt] coordinates {
   (5, 22.2) (10, 22.7) (15, 22.9) (20, 25.7) (25, 28.2) (30, 31.3) (35, 33.9) (40, 35.1)
};

 \end{axis}
\end{tikzpicture}
\caption{The EE versus the operational time of the miniature UAV-empowered holographic surfaces in the THz-enabled network.}\label{EE_vs_time}
\end{figure}

Fig.~\ref{EE_vs_time} offers a detailed examination of how the mission duration, represented by the miniature UAV's operational time $T$, affects EE across various benchmark schemes. The analysis reveals an interesting pattern: as mission time increases, there's a noticeable rise in EE for schemes utilizing fixed trajectories (Method D) and those starting with non-optimized but feasible configurations (`Initial'). This improvement in EE is due to extended communication opportunities and the ability to adjust flight parameters over time. However, this trend is not consistent across all methods; specifically, Methods A, B, and C do not show the same EE increase as $T$ grows.
Quantitatively, extending the mission duration results in EE improvements of at least \(\rm{37.1\%, 26.8\%, 22.8\%, 16.5\%,}\) and \(\rm{12.8\%}\) when using Methods A--E, respectively. These gains indicate that longer mission times provide a strategic benefit by allowing the holographic surfaces-assisted miniature UAV to optimize both communication metrics and flight adjustments, thereby enhancing the overall network QoS.

\begin{figure}[t]
\centering
\begin{tikzpicture}
\begin{axis}[
    width=0.49\textwidth,
    height=0.40\textwidth,
    xlabel={$M$ [Number of reflecting elements]},
    ylabel={$\eta_{EE}$ [$\rm{Mbits/Joule}$]},
    grid=major,
    legend style={
        at={(0.99,0.02)}, 
        anchor=south east, 
        font=\scriptsize,
        legend cell align={left},
        legend columns=2,
    },
    tick label style={font=\small},
    xlabel style={font=\small},
    ylabel style={font=\small},
    xmin=4, xmax=20,
    xtick={4,8,12,16,20},
    ymin=30, ymax=55,
    ytick={30,35,40,45,50,55},
    legend entries={
        \ \ \ \ \ \ \ Proposed,
        \!\!\!\!\!\!\!\!\!\!\!\!\!\!\!Solution,
        Method A,
        Method B,
        Method C,
        Method D,
        Method E,
        Initial
    }
]

\addplot[mark=pentagon, solid, black, line width=0.85pt] coordinates {
    (4, 47.2) (9, 49.0) (12, 50.5) (16, 51.8) (20, 52.9)
};
\addplot[solid, white, line width=0.85pt] coordinates {
    (1, 1)
};
\addplot[mark=o, solid, lightpurple, line width=0.85pt] coordinates {
    (4, 42.4) (9, 45.0) (12, 47.6) (16, 48.7) (20, 50.8)
};
\addplot[mark=10-pointed star, solid, red, line width=0.85pt] coordinates {
    (4, 39.5) (9, 42.7) (12, 44.2) (16, 46.1) (20, 49.0)
};
\addplot[mark=diamond, solid, darkgreen, line width=0.85pt] coordinates {
    (4, 38.5) (9, 41.6) (12, 43.5) (16, 45.7) (20, 48.5)
};
\addplot[mark=square, solid, blue, line width=0.85pt] coordinates {
    (4, 34.2) (9, 39.8) (12, 41.5) (16, 44.9) (20, 47.1)
};
\addplot[mark=triangle, solid, cyan, line width=0.85pt] coordinates {
    (4, 32.5) (9, 36.8) (12, 39.5) (16, 43.7) (20, 45.9)
};

\addplot[mark=+,densely dashdotted, brown, line width=0.85pt] coordinates {
    (4, 31.2) (9, 35.2) (12, 37.6) (16, 41.0) (20, 44.8)
};

\end{axis}
\end{tikzpicture}
\caption{\textcolor{black}{The impact of the number of reflecting elements, $M$, on the EE of THz-NOMA networks with a miniature UAV empowered by holographic surfaces.}}
\label{EE_RHSsize}
\end{figure}

\textcolor{black}{Fig.~\ref{EE_RHSsize} illustrates the variation in EE performance as the number of reflecting elements $M$ increases in THz-NOMA networks with a miniature UAV empowered by holographic surfaces. As expected, a larger $M$ enhances the system’s ability to manipulate the propagation environment, thereby improving both energy harvesting efficiency and link quality. This results in a steady growth in EE across all methods. However, the rate of improvement differs: baseline schemes such as Methods D and E show limited gains as $M$ increases, while Methods A$–$C benefit moderately from additional elements. In contrast, the proposed solution consistently achieves the highest EE values, with the performance gap widening as $M$ grows. This demonstrates that our joint optimization of UAV trajectory and NOMA power allocation is particularly effective when more reflecting elements are available, leveraging them to maximize harvested energy and improve transmission efficiency. These results highlight the scalability and robustness of the proposed design in exploiting large holographic surfaces for miniature UAV-assisted THz communication.}

\begin{figure}[t]
\centering
\begin{tikzpicture}
\begin{axis}[
    width=0.49\textwidth,
    height=0.40\textwidth,
    xlabel={Molecular absorption coefficient [m$^{-1}/1000$]},
    ylabel={$\eta_{EE}$ [Mbits/Joule]},
    grid=major,
    legend style={
        at={(0.99,0.02)}, 
        anchor=south east, 
        font=\scriptsize,
        inner sep=0.5mm, 
        legend cell align={left},
        legend columns=2,
        /tikz/column 1/.style={
            column sep=-1pt,
        },
        /tikz/column 2/.style={
            column sep=0pt,
        },
        /tikz/row 1/.style={
            row sep=-2pt,
        },
        /tikz/row 2/.style={
            row sep=-2pt, 
        },
        /tikz/row 3/.style={
            row sep=-5pt,
        },
    },
    tick label style={font=\small},
    xlabel style={font=\footnotesize, yshift=1.5mm},
    ylabel style={font=\footnotesize, yshift=-1.5mm},
    xmin=5, xmax=25,
    xtick={5,10,...,25},
    ymin=25, ymax=50,
    ytick={20,25,...,50},
    legend entries={
        Proposed Solution,
        Method A,
        $~~~~$Method B,
        Method C,
        $~~~~$Method D
    }
]

\addplot[mark=pentagon, solid, black, line width=0.7pt] coordinates {
    (5, 48.1) (10, 46.7) (15, 45.5) (20, 44.4) (25, 43.3) 
};
\addplot[mark=o, solid, lightpurple, line width=0.7pt] coordinates {
    (5, 44.5) (10, 43.0) (15, 41.5) (20, 40.7) (25, 39.2)
};
\addplot[mark=10-pointed star, solid, red, line width=0.7pt] coordinates {
    (5, 43.1) (10, 42.0) (15, 40.8) (20, 39.1) (25, 37.5)
};
\addplot[mark=diamond, solid, darkgreen, line width=0.7pt] coordinates {
    (5, 40.9) (10, 38.0) (15, 35.8) (20, 34.1) (25, 31.7)
};
\addplot[mark=square, solid, blue, line width=0.7pt] coordinates {
    (5, 35.1) (10, 33.0) (15, 31.8) (20, 28.1) (25, 25.5)
};

 \end{axis}
\end{tikzpicture}
\caption{\textcolor{black}{The impact of the molecular absorption coefficient on the EE of THz-NOMA networks with a miniature UAV empowered by holographic surfaces.}}
\label{EE_vs_mol}
\end{figure}

\textcolor{black}{Fig.~\ref{EE_vs_mol} illustrates the variation in EE performance as the molecular absorption coefficient for THz links changes according to environmental conditions, such as humidity and atmospheric composition. As expected, a higher molecular absorption coefficient leads to lower EE performance across all schemes, primarily due to increased propagation loss and attenuation caused by intense molecular absorption in the THz band. This effect not only reduces the amount of usable power available for energy harvesting at the miniature UAV but also degrades the signal quality received at the destination, thereby impairing the overall system efficiency. It is worth noting that such absorption effects are particularly critical in indoor or humid environments where water vapor dominates the THz channel response. Nevertheless, even under these adverse conditions, our proposed solution consistently outperforms the benchmark methods by jointly optimizing UAV trajectory and NOMA power allocation, thereby demonstrating its robustness and adaptability in realistic THz communication scenarios.}

Nevertheless, the relationship between mission time and EE is complex. The interplay among optimization variables produces a non-linear, though generally increasing, trend in EE as mission duration extends. This highlights the intricate dynamics of EE optimization, where certain adjustments can lead to substantial gains.
The observation that mission duration significantly impacts EE emphasizes an important challenge: \textit{minimizing the task completion time for miniature UAV-empowered holographic surfaces relay systems while meeting specific EE targets}. This requires balancing operational efficiency and mission urgency, suggesting a rich area for further research into optimizing UAV-based communication networks, with or without EH capabilities.

\section{Conclusion and Future Work}\label{sec_5}
In this paper, we explored the complexities of improving the efficiency of a cooperative THz NOMA-based miniature UAV network powered by EH holographic surfaces. We began by formulating an EE optimization problem aimed at refining the network’s resource allocation strategy. A novel deployment plan for the miniature UAV was introduced, designed to enhance THz wireless connectivity while accounting for molecular absorption effects, a crucial element in the path loss channel gain model for THz-enabled UAVs. Building on this, we developed an optimization framework to enhance EE, ensuring it met stringent QoS requirements. The optimization targeted key decision variables, including miniature UAV positioning and NOMA power allocation coefficients, based on a two-episode iterative solution. We demonstrated the effectiveness of the resource allocation algorithm through numerical results, highlighting its advantage when compared to baseline scenarios without trajectory and NOMA power optimizations. Our approach significantly improved network efficiency, extending UAV operational time and battery life, marking a major advancement in the capabilities of THz-NOMA miniature UAV networks with RHS.
\textcolor{black}{A potential direction for future research is to create a more realistic nonlinear EH model for RHS that incorporates active elements, enabling them to simultaneously \textit{absorb} energy for harvesting and \textit{reflect} signals for communication purposes.
Another exciting area would be optimizing the phase shifters of RHS to improve the energy absorption capabilities of each surface and, therefore, maximize EH.}



\appendices
{\color{black}
\section{Proof of the Proposition~\ref{prp_1}}\label{Appen_A}

Problems $\text{P}_1$ and $\text{P}_2$ are equivalent. The equivalence can be proven using an argument analogous to the one presented in [33] via a parametric transform of sum-of-ratios.  Let $x:=\{u[n]\}_{n=1}^N$ collect the trajectory variables in step one, given that power coefficients are fixed. Define
$f_n(x) \triangleq R_{\text{sum}}[n] \ge 0$ and $h_n(x) \triangleq \fged_\text{sum}[n] > 0$. Then the objective of $\text{P}_1$ is 
$\sum_{n=1}^N \frac{f_n(x)}{h_n(x)}$.

According to Section 2 of [33], by the parametric reformulation for sum-of-ratios maximization, there exist nonnegative parameters
$\{\beta_n\}_{n=1}^N$ and positive $\{u_n\}_{n=1}^N$ such that any optimizer $x^\star$ of $\text{P}_1$ also solves
\begin{subequations}
\begin{align}
& 
\underset{x}{\max} 
\: \:
\sum_{n=1}^N u_n\!\left(f_n(x)-\beta_n h_n(x)\right)
\nonumber
\\
\nonumber
&s.t.:  
\text{the constraints of } \text{P}_1
\end{align}
\end{subequations}
with the complementary relations
\begin{align}
& f_n(x^\star)-\beta_n h_n(x^\star)=0,\quad \forall n,\nonumber\\\nonumber
&  u_n=\frac{1}{h_n(x^\star)},\quad \forall n.
\end{align}
Setting $\alpha_n \triangleq u_n$, $f_n=R_{\text{sum}}[n]$, and $h_n=\fged_\text{sum}[n]$ yields exactly the subtractive
program $\text{P}_2$ with objective $\sum_{n=1}^N \alpha_n\!\left(R_{\text{sum}}[n]-\beta_n \fged_\text{sum}[n]\right)$ and the two
equalities
\begin{align}
& R_{\text{sum}}[n]-\beta_n \fged_\text{sum}[n]=0,\quad \forall n,\nonumber\\\nonumber
&  1-\alpha_n \fged_\text{sum}[n]=0,\quad \forall n.
\end{align}
which are Eqs.~\eqref{conn1} and ~\eqref{conn2}. Hence, for the fixed power coefficients, $\text{P}_1$ and $\text{P}_2$ share the same
optimizer $x^\star$. Conversely, any optimizer of $\text{P}_2$ that satisfies the two equalities attains the same objective
value as $\text{P}_1$. This establishes the claimed equivalence.
\QED
}

{\color{black}
\section{Proof of the Proposition~\ref{prp_2}}\label{Appen_B}
The proposed two-step iterative algorithm for solving problem $\text{P}_1$ 
monotonically improves the objective function value in each iteration 
and converges to a stationary point. 

\begin{proof}
Let $\{\boldsymbol{u}^{(j)}, \boldsymbol{\gimel}^{(j)}\}$ denote the feasible solution set 
of $\text{P}_1$ in the $j$-th iteration, 
where $\boldsymbol{u}$ represents the UAV trajectory and 
$\boldsymbol{\gimel}$ denotes the NOMA power allocation coefficients. 
At each iteration, two sub-problems are solved sequentially: 
1) UAV trajectory optimization ($\text{P}_5$),  
2) NOMA power allocation optimization ($\text{P}_7$).  

For a fixed $\boldsymbol{\gimel}^{(j)}$, solving $\text{P}_5$ yields 
$\boldsymbol{u}^{(j+1)}$ such that
\begin{align}
&f_{\text{P}_5}(\boldsymbol{u}^{(j+1)}, \boldsymbol{\gimel}^{(j)})
\geq f_{\text{P}_5}(\boldsymbol{u}^{(j)}, \boldsymbol{\gimel}^{(j)}),
\end{align}
where the inequality holds since the trajectory sub-problem is 
approximated via the SCA, which 
guarantees a non-decreasing objective sequence.  

Similarly, for a fixed $\boldsymbol{u}^{(j+1)}$, solving $\text{P}_7$ yields 
$\boldsymbol{\gimel}^{(j+1)}$ such that
\begin{align}
&f_{\text{P}_7}(\boldsymbol{u}^{(j+1)}, \boldsymbol{\gimel}^{(j+1)})
\geq f_{\text{P}_5}(\boldsymbol{u}^{(j+1)}, \boldsymbol{\gimel}^{(j)}),
\end{align}
where the inequality follows from the quadratic transform and 
fractional programming framework, which ensures monotonic 
improvement of the objective function.  

By combining the above, we obtain
\begin{align}
f_{\text{P}_7}(\boldsymbol{u}^{(j+1)}, \boldsymbol{\gimel}^{(j+1)})
\geq f_{\text{P}_5}(\boldsymbol{u}^{(j)}, \boldsymbol{\gimel}^{(j)}).
\end{align}
Thus, the sequence of objective function values is 
monotonically non-decreasing after each full iteration.  

Moreover, since the system EE objective is bounded above due to 
practical power and trajectory constraints (cf. constraints \eqref{1a}$-$\eqref{1c}, \eqref{power1}, \eqref{power2}, \eqref{p1_a}$-$\eqref{p1_e}), 
the iterative process must converge to a finite limit.  

Finally, by invoking the global optimization result in \cite[Theorem~3.1]{Jong2012AnEG}, 
which establishes that the equivalent sum-of-ratios formulation admits 
a unique solution under Lipschitz and strong monotonicity conditions, 
we conclude that the proposed algorithm converges to a stationary point 
with global linear and local superlinear/quadratic convergence rate. \QED
\end{proof}
}
\bibliographystyle{ieeetr}
\bibliography{ref}
\end{document}